\documentclass{llncs}
\usepackage{amsmath,amssymb}
\usepackage[dvipdf]{graphicx}
\usepackage[dvipdf]{color}
\usepackage{algorithmic,algorithm}

\newcommand{\msize}[1]{{\left|#1\right|}}

\newcommand{\bfI}{{I}}

\newcommand{\ini}{0}
\newcommand{\tar}{r}

\renewcommand{\algorithmicrequire}{\textbf{Input:}}
\renewcommand{\algorithmicensure}{\textbf{Output:}}

\newcommand{\Neiopen}[2]{N_{#1}(#2)}
\newcommand{\Neiclosed}[2]{N_{#1}[#2]}

\newcommand{\Ram}{{\sf Ramsey}}

\newcommand{\bfIp}{I^\prime}

\newcommand{\valbeta}{10k}
\newcommand{\valcli}{5k-2}

\newcommand{\bfIB}{\bfI^*}

\newcommand{\VR}{R}

\newcommand{\VN}{A}
\newcommand{\VB}{B}

\newcommand{\ON}{{\sf ON}}
\newcommand{\Nvec}[1]{\VN(#1)}

\newcommand{\cgraph}{\mathcal{C}}
\newcommand{\cvertex}{\mathcal{V}}
\newcommand{\cedge}{\mathcal{E}}

\newcounter{one}
\newcommand{\one}{{\rm \roman{one}}}
\newcounter{two}
\newcommand{\two}{{\rm \roman{two}}}
\newcounter{three}
\newcounter{four}
\newcounter{five}
\newenvironment{listing}[1]{%
        \begin{list}{*}{%
                 \settowidth{\labelwidth}{#1}%
                 \setlength{\leftmargin}{\labelwidth}%
                 \advance \leftmargin by 12pt
                   \setlength{\itemsep}{0pt}%
                   \setlength{\parsep}{0pt}%
                   \setlength{\topsep}{0pt}%
                   \setlength{\parskip}{0pt}%
}%
}{%
\end{list}}

\title{Fixed-Parameter Tractability \\ of Token Jumping on Planar Graphs}
\author{
Takehiro Ito\inst{1} 
\and 
Marcin Kami\'nski\inst{2} 
\and 
Hirotaka Ono\inst{3}  
}
\institute{%
        Graduate School of Information Sciences, 
		Tohoku University, \\
        Aoba-yama 6-6-05, Sendai, 980-8579, 
		Japan. \\
        \email{takehiro@ecei.tohoku.ac.jp}
\and
        Dept. of Mathematics, Computer Science and Mechanics, 
		University of Warsaw, \\
        Banacha 2, 02-097, Warsaw, 
		Poland. \\
    \email{mjk@mimuw.edu.pl}
\and
        Faculty of Economics, 
		Kyushu University,  \\
        Hakozaki 6-19-1, Higashi-ku, Fukuoka, 812-8581, 
		Japan. \\
	\email{hirotaka@econ.kyushu-u.ac.jp}
}

\begin{document}
\maketitle

\begin{abstract}
Suppose that we are given two independent sets $\bfI_{\ini}$ and $\bfI_{\tar}$ of a graph such that $\msize{\bfI_{\ini}}=\msize{\bfI_{\tar}}$, and imagine that a token is placed on each vertex in $\bfI_{\ini}$. 
The {\sc token jumping} problem is to determine whether there exists a sequence of independent sets of the same cardinality which transforms $\bfI_{\ini}$ into $\bfI_{\tar}$ so that each independent set in the sequence results from the previous one by moving exactly one token to another vertex. 
This problem is known to be PSPACE-complete even for planar graphs of maximum degree three, and W[1]-hard for general graphs when parameterized by the number of tokens. 
In this paper, we present a fixed-parameter algorithm for {\sc token jumping} on planar graphs, where the parameter is only the number of tokens. 
Furthermore, the algorithm can be modified so that it finds a shortest sequence for a yes-instance. 
The same scheme of the algorithms can be applied to a wider class of graphs which forbid a complete bipartite graph $K_{3,t}$ as a subgraph for a fixed integer $t\ge 3$.
% , and it yields fixed-parameter algorithms. 
\end{abstract}
%\begin{keywords}
%{fixed parameter tractability, graph algorithm, independent
%set, reachability on solution space} 
%\end{keywords}

	\section{Introduction}
	The {\sc token jumping} problem was introduced by Kami\'nski et al.~\cite{KaminskiMedvedevMilanic2012}, which can be seen as a ``dynamic'' version of independent sets in a graph. 
	An {\em independent set} of a graph $G$ is a set of vertices of $G$ in which no two vertices are adjacent. 
(See \figurename~\ref{fig:example}, which depicts six different independent sets of the same graph.) 
	Suppose that we are given two independent sets $\bfI_0$ and $\bfI_r$ of a graph $G = (V,E)$ such that $|\bfI_0| = |\bfI_r|$, and imagine that a token is placed on every vertex in $\bfI_0$. 
	Then, the {\sc token jumping} problem is to determine whether there exists a sequence $\langle \bfI_0, \bfI_1, \ldots, \bfI_\ell \rangle$ of independent sets of $G$ such that 
% the following conditions:    
	\begin{listing}{aaa}
	\item[(a)] $\bfI_{\ell} = \bfI_r$, and $|\bfI_0| = |\bfI_1| = \cdots =|\bfI_{\ell}|$; and 
	\item[(b)] for each index $i \in \{1,2, \ldots, \ell\}$, $\bfI_i$ can be obtained from $\bfI_{i-1}$ by moving exactly one token on a vertex $u \in \bfI_{i-1}$ to another vertex $v \not\in \bfI_{i-1}$, and hence $\bfI_{i-1} \setminus \bfI_i = \{u\}$ and $\bfI_i \setminus \bfI_{i-1} = \{v\}$.  
	\end{listing}
	Such a sequence is called a {\em reconfiguration sequence} between $\bfI_0$ and $\bfI_r$. 
	Figure~\ref{fig:example} illustrates a reconfiguration sequence $\langle \bfI_0, \bfI_1,\ldots, \bfI_5 \rangle$ of independent sets, which transforms $\bfI_0$ into $\bfI_r = \bfI_5$;
therefore, the answer is ``YES'' for this instance. 
\begin{figure}[t]
\begin{center}
\includegraphics[width=0.6\textwidth]{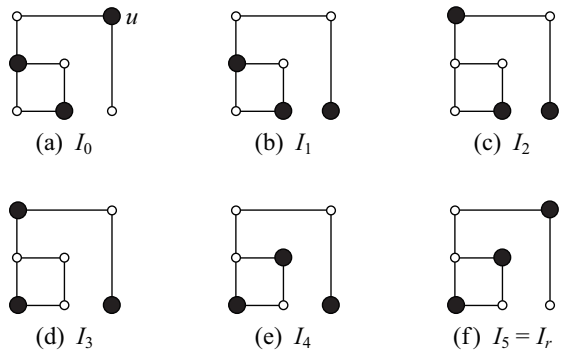}
\end{center}
\vspace{-1em}
\caption{A sequence $\langle \bfI_{\ini}, \bfI_1, \ldots, \bfI_5 \rangle$ of independent sets of the same graph, where the vertices in independent sets are depicted by large black circles (tokens).}
\label{fig:example}
\end{figure}

	Recently, similar settings of problems have been extensively studied in the framework of reconfiguration problems~\cite{IDHPSUU}, which arise when we wish to find a step-by-step transformation between two feasible solutions of a problem instance such that all intermediate solutions are also feasible and each step abides by a prescribed reconfiguration rule 
(i.e., an adjacency relation defined on feasible solutions of the original problem). 
	For example, the {\sc token jumping} problem can be seen as a reconfiguration problem for the (ordinary) independent set problem:
feasible solutions are defined to be all independent sets of the same cardinality in a graph, as in the condition~(a) above; 
and the reconfiguration rule is defined to be the condition~(b) above. 
	This reconfiguration framework has been applied to several well-studied combinatorial problems, including 
independent set~\cite{Bon14,BKW14,HearnDemaine2005,IDHPSUU,ItoKaminskiOnoSuzukiUeharaYamanaka2014,KaminskiMedvedevMilanic2012,MNRSS13,MNRW14,Wro14},
satisfiability~\cite{Kolaitis}, 
set cover, clique, matching~\cite{IDHPSUU}, 
vertex-coloring~\cite{BJLPP14,BC09}, 
list edge-coloring~\cite{IKD09},
(list) $L(2,1)$-labeling~\cite{IKOZ_isaac}, 
% subset sum~\cite{ID11},  
% shortest path~\cite{KMP11}, 
and so on.

	\subsection{Known and related results}

	The first reconfiguration problem for independent set, called {\sc token sliding}, was introduced by Hearn and Demaine~\cite{HearnDemaine2005} which employs another reconfiguration rule. 
	Indeed, there are three reconfiguration problems for independent set, called 
{\sc token jumping}~\cite{BKW14,ItoKaminskiOnoSuzukiUeharaYamanaka2014,KaminskiMedvedevMilanic2012,Wro14}, 
{\sc token sliding}~\cite{BC09,BKW14,HearnDemaine2005,KaminskiMedvedevMilanic2012,Wro14}, and 
{\sc token addition and removal}~\cite{Bon14,IDHPSUU,KaminskiMedvedevMilanic2012,MNRSS13,MNRW14,Wro14}. 
(See \cite{KaminskiMedvedevMilanic2012} for the definitions.)
	These are the most intensively studied reconfiguration problems, and hence
% , as the number of references indicated. 
we here explain only the results strongly related to this paper; see the references above for the other results. 

	First, {\sc token jumping} (indeed, all three reconfiguration problems for independent set) is PSPACE-complete for planar graphs of maximum degree three~\cite{BC09,HearnDemaine2005,ItoKaminskiOnoSuzukiUeharaYamanaka2014}, for perfect graphs~\cite{KaminskiMedvedevMilanic2012}, and for bounded bandwidth graphs~\cite{Wro14}.

	Second, Kami\'nski et al.~\cite{KaminskiMedvedevMilanic2012} gave a linear-time algorithm for {\sc token jumping} on even-hole-free graphs. 
	Furthermore, their algorithm can find a reconfiguration sequence with the shortest length.
%  (i.e., the minimum number of token movements).

	Third, Ito et al.~\cite{ItoKaminskiOnoSuzukiUeharaYamanaka2014} proved that {\sc token jumping} is W[1]-hard for general graphs when parameterized only by the number of tokens. 
	Therefore, it is very unlikely that the problem admits a fixed-parameter algorithm for general graphs when the parameter is only the number of tokens. 
	They also gave a fixed-parameter algorithm for general graphs when parameterized by both the number of tokens and the maximum degree of graphs. 
	Their algorithm can be modified so that it finds a reconfiguration sequence with the shortest length. 

\if0
\medskip

	As the third known result above, reconfiguration problems have been studied recently under the parameterized complexity framework:
for example, 
{\sc token addition and removal}~\cite{MNRSS13,MNRW14},
reconfiguration problems for 
% vertex cover~\cite{MNRSS13,MNRW14},
feedback vertex set~\cite{MNRSS13,MNRW14}, 
vertex-coloring~\cite{BM14,JKKPP14},
and so on. 
	However, almost all of the known results take the length $\ell$ of a reconfiguration sequence as the parameter. 
	This is a certainly natural choice, but unfortunately the length parameter $\ell$ affects the reconfigurability (i.e., the existence/nonexistence of a reconfiguration sequence).  
	For example, \figurename~\ref{fig:example} is a yes-instance whose shortest reconfiguration sequence is of length six. 
	However, it would be a no-instance (in the parameterized problem) if the length parameter $\ell$ is set by less than six. 

% 	Therefore, it is desirable to develop a fixed-parameter algorithm for {\sc token jumping} without choosing the length of a reconfiguration sequence as the parameter. 
\fi

	\subsection{Our contribution}
	In this paper, we first give a fixed-parameter algorithm for {\sc token jumping} on planar graphs when parameterized only by the number of tokens. 
% 	Therefore, this algorithm can always determine the reconfigurability between two given independent sets, irrespective of the length of a reconfiguration sequence. 
	Recall that {\sc token jumping} is PSPACE-complete for planar graphs of maximum degree three, and is W[1]-hard for general graph when the parameter is only the number of tokens. 

% 	It is well known that planar graphs are $K_5$-minor free and $K_{3,3}$-minor free~\cite{BLS99}. 
	Interestingly, our algorithm for planar graphs utilizes only the property that no planar graph contains a complete bipartite graph $K_{3,3}$ as a subgraph~\cite{BLS99}. 
	We show that the same scheme of the algorithm for planar graphs can be applied to a wider class of graphs which forbid a complete bipartite graph $K_{3,t}$ as a subgraph for a fixed integer $t\ge 3$.
% , and it yields a fixed-parameter algorithm. 
(We call such graphs {\em $K_{3,t}$-forbidden graphs}.)

	In addition, the algorithm for $K_{3,t}$-forbidden graphs (and hence for planar graphs) can be modified so that it finds a reconfiguration sequence with the shortest length
%  (i.e., the minimum number of token movements) 
for a yes-instance. 
	We note that the reconfiguration sequence in \figurename~\ref{fig:example} is shortest. 
	It is remarkable that the token on the vertex $u$ in \figurename~\ref{fig:example}(a) must make a ``detour'' to avoid violating the independence of tokens: it is moved twice even though $u \in \bfI_{\ini} \cap \bfI_{\tar}$. 
	Our algorithm can capture such detours for $K_{3,t}$-forbidden graphs. 

% 	Finally, we will give a simple observation that our fixed-parameter algorithms also solve {\sc ISReconf} under the TAR rule for $K_{3,t}$-forbidden graphs for any fixed integer $t\ge 3$. 

	\subsection{Strategy for fixed-parameter algorithms}

	We here explain two main ideas to develop a fixed-parameter algorithm for {\sc token jumping}; 
formal descriptions will be given later. 

	The first idea is to find a sufficiently large ``buffer space'' to move the tokens. 
	Namely, we first move all the tokens from $\bfI_0$ to the buffer space, and then move them from the buffer space to $\bfI_r$;
thus, the answer is ``YES'' if we can find such a buffer space. 
% , which always guarantees the existence of a reconfigurable sequence. 
	Due to the usage, such a buffer space (a set of vertices) should be mutually independent and preferably not adjacent to any vertex in $\bfI_0 \cup \bfI_r$.  

	The second idea is to ``shrink the graph'' into a smaller one with preserving the reconfigurability (i.e., the existence/nonexistence of a reconfiguration sequence) between $\bfI_0$ and $\bfI_r$.
	This idea is based on the claim that, if the size of the graph is bounded by a function depending only on the parameter $k$, we can solve the problem in a brute-force manner in fixed-parameter running time. 
	Thus, it is useful to find such ``removable'' vertices in fixed-parameter running time, and shrink the graph so that the size of the resulting graph is bounded by a function of $k$.
% ; 
% then we can solve the problem in fixed-parameter running time in total. 

	The $K_{3,t}$-forbiddance (and hence $K_{3,3}$-forbiddance) of graphs satisfies the two main ideas above at the same time:
it ensures that the graph has a sufficiently large independent sets, which may be used as a buffer space; and 
it characterizes removable vertices. 

\smallskip
Due to the page limitation, we omit some proofs from this extended abstract.

	\section{Preliminaries}

	In this paper, we assume without loss of generality that graphs are simple. 
	Let $G = (V,E)$ be a graph with vertex set $V$ and edge set $E$. 
	The {\em order} of $G$ is the number of vertices in $G$. 
	We say that a vertex $w$ in $G$ is a {\em neighbor} of a vertex $v$ if $\{v,w\} \in E$. 
	For a vertex $v$ in $G$, let $\Neiopen{G}{v} = \{w \in V \mid \{v, w\} \in E\}$.
	We also denote $\Neiopen{G}{v} \cup \{v\}$ by $\Neiclosed{G}{v}$. 
	For a vertex set $S \subseteq V$, let $\Neiopen{G}{S} = \bigcup_{v \in S} \Neiopen{G}{v}$ and $\Neiclosed{G}{S} = \bigcup_{v \in S} \Neiclosed{G}{v}$. 
\smallskip

	For a vertex set $S \subseteq V$ of a graph $G=(V,E)$, $G[S]$ denotes the subgraph {\em induced} by $S$, that is, $G[S]=(S,E[S])$ where $E[S]=\{ \{u,v\} \in E \mid \{u,v\} \subseteq S\}$. 
	A vertex set $S$ of $G$ is an {\em independent set} of $G$ if $G[S]$ contains no edge. 
	A subgraph $G^\prime$ of $G$ is called a {\em clique} if every pair of vertices in $G^\prime$ is joined by an edge;
we denote by $K_s$ a clique of order $s$. 
% 	A graph $G = (V,E)$ is {\em bipartite} if $V$ can be partitioned into $X$ and $Y$ such that both $G[X]$ and $G[Y]$ form independent sets of $G$.  
	For two positive integers $p$ and $q$, we denote by $K_{p,q}$ a complete bipartite graph with its bipartition of size $p$ and $q$. 
% can be partitioned into two subsets $P$ and $Q$ such that $G[P]$ and $G[Q]$ form independent set of $G$ and every vertex in $P$ is 

% For a given graph $G=(V,E)$, a sequence of the vertices
% $v_0,v_1,\cdots,v_l$ is a {\em path}, denoted by $(v_0,v_1,\cdots,v_l)$,
% if $\{v_j,v_{j+1}\}\in E$ for each $0\le j\le l-1$. The {\em length} of
% a path is the number of edges on the path. For two independent sets
% $\bfI_a$ and $\bfI_b$ in a graph $G=(V,E)$, we sometimes consider a
% transition between them such that $\bfI_{a}\setminus\bfI_{b}=\{u\}$ and
% $\bfI_{b}\setminus\bfI_{a}=\{v\}$ as a {\em sliding token} from $u$ to
% $v$ along the edge $\{u,v\}$. 
% If $\bfI_{b}$ can be obtained from $\bfI_{a}$ by sliding a token along an edge,
% we denote it by $\bfI_{a}\vdash \bfI_{b}$.
% Similarly, we sometimes slide a token from a vertex to another vertex
% along a path, and denote it by $\vdash^*$.
	\medskip

% 	\noindent
% 	{\bf Planar graphs.}
	
	A graph is {\em planar} if it can be embedded in the plane without any edge-crossing~\cite{BLS99}. 
	In Section~\ref{sec:algorithm}, our algorithms utilize an independent set of sufficiently large size in a graph, as a buffer space to move tokens. 
	As for independent sets of planar graphs, the following is known, though the original description is about the four-color theorem.
	\begin{proposition}[\cite{robertson1996efficiently}]
	\label{prop:4color}
	For a planar graph of order $n=4s$, there exists an independent set of size at least $s$, and it can be found in $O(n^2)$ time.  
	\end{proposition}

% 	\medskip
% 	
% 	\noindent
% 	{\bf $K_{3,t}$-free graphs.}

% 	It is well known as Kuratowski's theorem that a graph is planar if and only if it does not contain $K_5$ or $K_{3,3}$ as a minor. 
	It is well known as Kuratowski's theorem that a graph is planar if and only if it does not contain a subdivision of $K_5$ or $K_{3,3}$~\cite{BLS99}. 
% (Namely, planar graph is a typical minor-closed graph class.)
	Therefore, any planar graph contains neither $K_5$ nor $K_{3,3}$ as a subgraph.
%  that is, it is $K_5$-free and $K_{3,3}$-free. 
	In this paper, we extend our algorithm for planar graphs to a much larger class of graphs.
% , $K_{3,t}$-free graphs for any fixed integer $t \ge 3$.
	For two positive integers $p$ and $q$, a graph is {\em $K_{p,q}$-forbidden} if it contains no $K_{p,q}$ as a subgraph.
	For example, any planar graph is $K_{3,3}$-forbidden. 
% 	Note that any $K_{p,q}$-free graph is $K_{p,q}$-minor free, and hence the class of $K_{p,q}$-free graphs contains that of $K_{p,q}$-minor-free graphs properly.  
	It is important that any $K_{p,q}$-forbidden graph contains no clique $K_{p+q}$ of size $p+q$. 
	
	In our algorithm for $K_{3,t}$-forbidden graphs in Section~\ref{subsec:extension}, we use Ramsey's theorem, instead of Proposition~\ref{prop:4color}, to guarantee a sufficiently large independent set.
	Ramsey's theorem states that, for every pair of integers $a$ and $b$, there exists an integer $n$ such that any graph of order at least $n$ has an independent set of size $a$ or a clique of size $b$ 
(see~\cite{Graham:1990:RT:85271} for example). 
	The smallest number $n$ of vertices required to achieve this property is called a {\em Ramsey number}, denoted by $\Ram(a,b)$. 
	It is known that $\Ram(a,b)\le {a+b-2 \choose b-1}$~\cite{Graham:1990:RT:85271}.  
	Since any $K_{p,q}$-forbidden graph contains no $K_{p+q}$, we have the following proposition. 
	\begin{proposition}
	\label{prop:free}
	For integers $p, q$ and $s$, let $G$ be a $K_{p,q}$-forbidden graph of order at least $\Ram(s, p+q)$.
	Then, $G$ has an independent set of size at least $s$. 
	\end{proposition}

%%%%%%%%%%%%%%%%%%%
\section{Fixed-Parameter Algorithm}
\label{sec:algorithm}
%%%%%%%%%%%%%%%%%%%
	In this section, we present a fixed-parameter algorithm for planar graphs to determine if a given {\sc token jumping} instance is reconfigurable or not, as in the following theorem. 
% It is a fixed-parameter algorithm, and the parameter is the number $k$ of tokens. 
% Formally, we give the following theorem. 
\begin{theorem} \label{the:planar}
{\sc Token jumping} with $k$ tokens can be solved for planar graphs $G=(V,E)$ in
 $O \bigl(|E|+ \bigl(f_1(k) \bigr)^{2k} \bigr)$ time, where $f_1(k) = 2^{6k+1} + 180k^3$. 
\end{theorem}

As a proof of Theorem~\ref{the:planar}, we will prove that Algorithm~\ref{alg}, described below, is such an algorithm.
In Section~\ref{subsec:planar}, we will explain the algorithm step by step, together with its correctness. 
We will show in Section~\ref{subsec:extension} that our algorithm for planar graphs can be extended to that for $K_{3,t}$-forbidden graphs, $t \ge 3$.
% where $t \ge 3$ is any fixed constant. 
% In fact, Algorithm \ref{alg} uses not the planarity but the $K_{3,3}$-freeness only. 

\subsection{Planar graphs} 
\label{subsec:planar}

	As we have mentioned in Introduction, our algorithm is based on two main ideas: 
it returns ``YES'' as soon as we can find a sufficiently large buffer space (Lemmas~\ref{lem:buffer} and \ref{lem:step2});
otherwise it shrinks the graph so as to preserve the existence/nonexistence of a reconfiguration sequence between two given independent sets $\bfI_0$ and $\bfI_r$ (Lemma~\ref{lem:shrinkreconf}).
	After shrinking the graph into a smaller one of the order depending only on $k$, we can solve the problem in a brute-force manner (Lemma~\ref{lem:brute}). 
	It is important to notice that our algorithm returns ``NO'' only in this brute-force step.  
	In the following, we explain how the algorithm finds a buffer space or shrinks the graph, which well utilizes the $K_{3,3}$-forbiddance of $G$. 
% \smallskip

At the beginning part of the algorithm (lines~\ref{alg:s0}--\ref{alg:f0}), we set two
parameters $\alpha$ and $\beta$ as $4k$ and $\valbeta$, respectively. 
These are the orders of (sub)graphs that guarantee the existence of
sufficiently large independent sets that will be used as a buffer space.
%  of sizes $k$ and $\valcli$, respectively (recall Proposition~\ref{prop:4color}). 
	Let $\VN = \Neiopen{G}{\bfI_0 \cup \bfI_r} \setminus (\bfI_0 \cup \bfI_r)$, that is, the set of vertices that are not in $\bfI_0 \cup \bfI_r$ and have at least one neighbor in $\bfI_0 \cup \bfI_r$.
	Let $\VR =V\setminus \Neiclosed{G}{\bfI_0 \cup \bfI_r}$.
	Then, no vertex in $\VR$ is adjacent with any vertex in $\bfI_0 \cup \bfI_r$. 
	Notice that $\bfI_0 \cup \bfI_r$, $\VN$ and $\VR$ form a partition of $V$. 
\smallskip

\begin{algorithm}[t]
	\caption{TokenJump for planar graphs}\label{alg}
	\begin{algorithmic}[1]
	\REQUIRE A parameter $k$, a planar graph $G=(V,E)$, and two independent sets $\bfI_0$ and $\bfI_r$ of $G$ such that $|\bfI_0| = |\bfI_r| = k$.
	\ENSURE ``YES'' if there is a reconfiguration sequence between $\bfI_0$ and $\bfI_r$; otherwise ``NO.'' 
	\STATE $\alpha:= 4k$, $\beta:=\valbeta$. \label{alg:s0}
	\STATE $\VN := \Neiopen{G}{\bfI_0 \cup \bfI_r} \setminus (\bfI_0 \cup \bfI_r)$, $\VR :=V \setminus \Neiclosed{G}{\bfI_0 \cup \bfI_r}$. \label{alg:f0}
	\IF[Step~1: $\VR$ has a sufficiently large buffer space]{$|\VR| \ge \alpha$} \label{alg:s1}
		\RETURN ``YES'' and exit. \label{alg:f1}
	\ELSE[$|\VR| < \alpha$ holds below]
		\FOR{each vector $\vec{x}\in \{0,1\}^{V \setminus \VN}$} \label{alg:sfor}
			\STATE $\Nvec{\vec{x}} := \{v\in \VN \mid \Neiopen{G}{v} \cap (V \setminus \VN) = \ON(\vec{x})\}$.
			\IF{$|\Nvec{\vec{x}}| \ge \beta$}
				\IF{$|\ON(\vec{x}) \cap \bfI_0| \le 1$ and $|\ON(\vec{x}) \cap \bfI_r| \le 1$} \label{alg:s2}
					\STATE \COMMENT{Step~2: $\Nvec{\vec{x}}$ has a sufficiently large buffer space}
					\RETURN ``YES'' and exit. \label{alg:f2}
				\ELSE[Step~3: shrink the graph]
					 \STATE Choose an arbitrary subset $\VB(\vec{x})$ of $\Nvec{\vec{x}}$ with $\beta$ vertices, and remove all vertices in $\Nvec{\vec{x}} \setminus \VB(\vec{x})$ from $V$ (and update $V$). \label{alg:s3}
				\ENDIF
			\ENDIF 
		\ENDFOR \COMMENT{$|\Nvec{\vec{x}}| \le \beta$ hold for all vectors $\vec{x}\in \{0,1\}^{V \setminus \VN}$} \label{alg:ffor}
	\ENDIF \COMMENT{The order of $G$ now depends only on $k$} \label{alg:f5}
	\STATE Check the existence of a reconfiguration sequence in a brute-force manner. \label{alg:s4}
\end{algorithmic}
\end{algorithm}

\noindent
{\bf Step 1}: Lines~\ref{alg:s1}--\ref{alg:f1} of Algorithm \ref{alg}.

If $|\VR| \ge \alpha = 4k$, then by Proposition~\ref{prop:4color} the subgraph $G[\VR]$ has an independent set of size at least $k$. 
Then, the algorithm returns ``YES'' because we can use it as a buffer space, as follows. 
	\begin{lemma} \label{lem:buffer}
	If $|\VR| \ge \alpha$, there is a reconfiguration sequence between $\bfI_0$ and $\bfI_r$.
	\end{lemma}
	\begin{proof}
	Let $\bfIp$ be an independent set of $G[\VR]$ with $|\bfIp| \ge k$; by Proposition~\ref{prop:4color} such an independent set $\bfIp$ always exists. 
	Since no vertex in $G[\VR]$ is adjacent with any vertex in $\bfI_0 \cup \bfI_r$, there is a reconfiguration sequence between $\bfI_0$ and $\bfI_r$ via $\bfIp$, as follows:
move all tokens on the vertices in $\bfI_0$ to vertices in $\bfIp$ one by one; 
and move all tokens on vertices in $\bfIp$ to the vertices in $\bfI_r$ one by one. 
	\qed
	\end{proof}
% Therefore, the algorithm returns ``YES'' if $|\VR| \ge \alpha = 4k$.
% \medskip

\noindent
{\bf Step 2}: Lines~\ref{alg:s2}--\ref{alg:f2} of Algorithm \ref{alg}.

We now know that $|\VR| < \alpha$. 
Since $\VR$ was small, the algorithm then tries to find a sufficiently large buffer space in $\VN$.
Notice that 
	\begin{equation} \label{eq:step2}
		|V \setminus \VN| = |\bfI_0 \cup \bfI_r \cup \VR| < 2k + \alpha,
	\end{equation}
which depends only on $k$. 
We will partition $\VN$ into at most $2^{2k+\alpha} = 2^{6k}$ subsets, according to how the vertices in $\VN$ are adjacent with vertices in $V\setminus \VN$.
%  $\VNone$ and $\VNtwo$, that will be defined later;
% this step deals with $\VNone$, and Step~3 will deal with $\VNtwo$. 

Before partitioning $\VN$, we first introduce a new notation. 
For a vertex set $S \subseteq V$, let $\vec{x}$ be an $|S|$-dimensional binary vector in $\{0,1\}^{S}$;
% , where each variable corresponds to a vertex in $S$. For $v\in S$, 
we denote by $x_v$ the component of $\vec{x}$ corresponding to a vertex $v \in S$. 
For each vector $\vec{x}\in \{0,1\}^S$, let $\ON(\vec{x}) = \{v\in S \mid x_v = 1\}$. 
For example, if $\vec{x}=(10011)\in \{0,1\}^{\{1,4,5,6,8\}}$ for a vertex set $S = {\{1,4,5,6,8\}}$, then $x_1=1$, $x_4=0$, $x_5=0$, $x_6=1$, $x_8=1$ and $\ON(\vec{x})=\{1,6,8\}$. 

\begin{figure}[t]
\begin{center}
\includegraphics[width=0.45\textwidth]{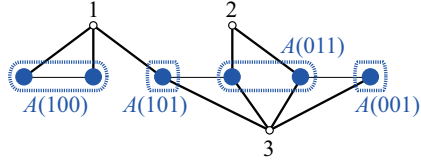}
\end{center}
\vspace{-2em}
\caption{Partitioning $\VN$ of six large (blue) vertices into four subsets $\Nvec{100}$, $\Nvec{101}$, $\Nvec{011}$ and $\Nvec{001}$, where each vertex in $V \setminus \VN = \{1,2,3\}$ is represented by a small (white) circle.}
\label{fig:Aexample}
\end{figure}

To partition the vertex set $\VN$, we prepare all binary vectors in $\{0,1\}^{V\setminus \VN}$. 
% Note that at this moment $|V\setminus N_2|=|I_0\cup I_r \cup N_1 \cup R|\le 2k+2\alpha=10k$ holds. 
By Eq.~(\ref{eq:step2}) the number of the prepared vectors is at most $2^{2k+\alpha} = 2^{6k}$. 
For each vector $\vec{x}\in \{0,1\}^{V \setminus \VN}$, we define $\Nvec{\vec{x}} = \{v\in \VN \mid \Neiopen{G}{v} \cap (V \setminus \VN) = \ON(\vec{x})\}$,  that is,    
$\vec{x}$ is used to represent a pattern of neighbors in $V \setminus \VN$. 
(See \figurename~\ref{fig:Aexample}.)
Therefore, all vertices in the same subset $\Nvec{\vec{x}}$ have exactly the same neighbors in $V \setminus \VN = \bfI_0 \cup \bfI_r \cup \VR$, that is, the vertices in $\ON(\vec{x})$. 
Conversely, each vertex in $\ON(\vec{x})$ is adjacent with {\em all} vertices in $\Nvec{\vec{x}}$. 
We thus have the following proposition. 
	\begin{proposition} \label{prop:K}
	For each vector $\vec{x}\in \{0,1\}^{V \setminus \VN}$, $G[\Nvec{\vec{x}} \cup \ON(\vec{x})]$ contains a complete bipartite graph $K_{|\Nvec{\vec{x}}|, |\ON(\vec{x})|}$ as a subgraph whose bipartition consists of $\Nvec{\vec{x}}$ and $\ON(\vec{x})$.
	\end{proposition}
Note that 
% some patterns (vectors) are actually not used, For example, 
$0$-vector (i.e., every component is $0$) is not used, because each vertex in $\VN$ is adjacent to at least one vertex in $\bfI_0 \cup \bfI_r$. 
% Also, $\vec{x}$ of $|ON(\vec{x})|=1$ is 
% not used, because a vertex that is adjacent only 
% to a vertex in $I_0\cup I_r$ is classified to $N_1$. Thus, meaningful 
% vectors satisfy $|ON(\vec{x})|\ge 2$.  
In this way, we partition $\VN$ into at most $2^{2k+\alpha}$ subsets $\Nvec{\vec{x}}$ according to the vectors $\vec{x} \in \{0,1\}^{V\setminus \VN}$.
	Proposition~\ref{prop:K} and the $K_{3,3}$-forbiddance give the following property on $\ON(\vec{x})$.
(Note that $\beta \ge 3$.) 
	\begin{lemma} \label{lem:atmosttwo}
	If $|\Nvec{\vec{x}}| \ge \beta$ holds for a vector $\vec{x} \in \{0,1\}^{V\setminus \VN}$, then $|\ON(\vec{x})| \le 2$. 
	\end{lemma}
% 	\begin{proof}
% 	Suppose for a contradiction that $|\ON(\vec{x})| \ge 3$ holds, and hence at least three distinct vertices are contained in $\ON(\vec{x})$.
% 	By the definition, each of them is adjacent with all vertices in $\Nvec{\vec{x}}$. 
% 	Since $|\Nvec{\vec{x}}| \ge \beta = 8k + 4 > 3$, we then have $K_{3,3}$ as a subgraph of $G$.
% 	This contradicts the $K_{3,3}$-freeness of $G$. 
% 	\qed
% 	\end{proof}

	The algorithm tries to find a sufficiently large buffer space from one of the subsets $\Nvec{\vec{x}}$, $\vec{x} \in \{0,1\}^{V\setminus \VN}$, such that $|\Nvec{\vec{x}}| \ge \beta = \valbeta$.
	The following lemma proves the correctness of Step~2.
	\begin{lemma} \label{lem:step2}
	Suppose that there exists a binary vector $\vec{x} \in \{0,1\}^{V\setminus \VN}$ such that $|\Nvec{\vec{x}}| \ge \beta$, $|\ON(\vec{x}) \cap \bfI_0| \le 1$ and $|\ON(\vec{x}) \cap \bfI_r| \le 1$. 
	Then, there exists a reconfiguration sequence between $\bfI_0$ and $\bfI_r$. 
	\end{lemma}
	\begin{proof}
	Suppose that such a vector $\vec{x}$ exists.
	Since $|\Nvec{\vec{x}}| \ge \beta = \valbeta$, by Proposition~\ref{prop:4color} the graph $G[\Nvec{\vec{x}}]$ has an independent set $\bfIp$ of size at least $2k$. 
	Let $w_0 \in \ON(\vec{x}) \cap \bfI_0$ and $w_r \in \ON(\vec{x}) \cap \bfI_r$ if such vertices exist;
$w_0 = w_r$ may hold.
% (See \figurename~\ref{fig:step2}.)
	Then, we obtain a reconfiguration sequence between $\bfI_0$ and $\bfI_r$ via $\bfIp$, as follows:
		\begin{listing}{aaa}
		\item[(a)] move the token on $w_0$ to an arbitrary vertex $w^\prime$ in $\bfIp$;
		\item[(b)] move all tokens in $\bfI_0 \setminus \{w_0 \}$ to vertices in $\bfIp \setminus \{w^\prime\}$ one by one; 
		\item[(c)] move tokens in $\bfIp \setminus \{w^\prime\}$ to the vertices in $\bfI_r \setminus \{w_r\}$ one by one; and 
		\item[(d)] move the last token on $w^\prime \in \bfIp$ to $w_r$.
		\end{listing}
	Note that no vertex in $\bfI_0 \setminus \{w_0\}$ is adjacent with any vertex in $\bfIp$ because $|\ON(\vec{x}) \cap \bfI_0| \le 1$.
	Furthermore, since $\bfIp$ is an independent set in $G[\Nvec{\vec{x}}]$, $w^\prime \in \bfIp$ is not adjacent with any vertex in $\bfIp \setminus \{w^\prime\}$.  
	Therefore, we can execute both (a) and (b) above without violating the independence of tokens. 
	By the symmetric arguments, we can execute both (c) and (d) above, too. 
	Thus, according to (a)--(d) above, there exists a reconfiguration sequence between $\bfI_0$ and $\bfI_r$.
	\qed
	\end{proof}

% \begin{figure}[t]
% \begin{center}
% \includegraphics[width=0.75\textwidth]{figure/step2.eps}
% \end{center}
% \vspace{-2em}
% \caption{Illustration for Lemma~\ref{lem:step2}.}
% \label{fig:step2}
% \end{figure}

% 	Thus, in Lines~\ref{alg:s2}--\ref{alg:f2}, the algorithm returns ``YES'' if there exists a subset $\Nvec{\vec{x}}$ such that $|\Nvec{\vec{x}}| \ge \beta$, $|\ON(\vec{x}) \cap \bfI_0| \le 1$ and $|\ON(\vec{x}) \cap \bfI_r| \le 1$. 
% \medskip

\noindent
{\bf Step 3}: Line~\ref{alg:s3} of Algorithm \ref{alg}.

We now consider to shrink the graph:
the algorithm shrinks each subset $\Nvec{\vec{x}}$ of size more than $\beta$ into a smaller one $\VB(\vec{x})$ of size $\beta$. 
% In this step, we do not care the subsets $\Nvec{\vec{x}}$ of size at most $\beta$ from the first.

Consider any subset $\Nvec{\vec{x}}$ of size more than $\beta$. 
Then, by Lemma~\ref{lem:atmosttwo} we have $|\ON(\vec{x})| \le 2$. 
In fact, since we have executed Step~2, $|\ON(\vec{x})\cap \bfI_0|= 2$ or $|\ON(\vec{x})\cap \bfI_r|= 2$ holds (recall Lemma~\ref{lem:step2}). 
We choose an arbitrary set $\VB(\vec{x})$ of $\beta = \valbeta$ vertices from $\Nvec{\vec{x}}$. 
Then, localizing independent sets intersecting $\Nvec{\vec{x}}$ only to $\VB(\vec{x})$ does not affect the reconfigurability, as in the following lemma. 
	\begin{lemma} \label{lem:shrinkreconf}
	$G$ has a reconfiguration sequence between $\bfI_0$ and $\bfI_r$ if and only if there exists a reconfiguration sequence $\langle \bfI_0, \bfIp_1, \bfIp_2, \ldots, \bfIp_{\ell^\prime}, \bfI_r \rangle$ such that $\bfIp_{j}\cap \Nvec{\vec{x}} \subseteq \VB(\vec{x})$ holds for every index $j \in \{1, 2, \ldots, \ell^\prime\}$. 
	\end{lemma}
	\begin{proof}
	The if-part clearly holds, and hence we prove the only-if-part.
% 	We may assume that $|\ON(\vec{x}) \cap \bfI_0| = 2$; 
% it is symmetric for the case where $|\ON(\vec{x}) \cap \bfI_r| = 2$. 
% 	Then, in the initial independent set $\bfI_0$, two tokens are placed on the two common neighbors in $\ON(\vec{x}) \cap \bfI_0$ of $\Nvec{\vec{x}}$.
%
	Suppose that $G$ has a reconfiguration sequence $\langle \bfI_0, \bfI_1, \ldots, \bfI_{\ell}, \bfI_r \rangle$ between $\bfI_0$ and $\bfI_r$. 
	If $\bfI_{j} \cap \Nvec{\vec{x}} \subseteq \VB(\vec{x})$ holds for every $j \in \{1, 2, \ldots, \ell \}$, the claim is already satisfied. 
	Thus, let $\bfI_{p}$ and $\bfI_{q}$ be the first and last independent sets of $G$, respectively, in the subsequence $\langle \bfI_1, \bfI_2, \ldots, \bfI_{\ell} \rangle$ such that $\bfI_j \cap (\Nvec{\vec{x}} \setminus \VB(\vec{x})) \neq \emptyset$.
	Note that $p = q$ may hold.
	Then, $\bfI_p$ contains exactly one vertex $w_p$ in $\bfI_p \cap (\Nvec{\vec{x}} \setminus \VB(\vec{x}))$, and $\bfI_q$ contains exactly one vertex $w_q$ in $\bfI_q \cap (\Nvec{\vec{x}} \setminus \VB(\vec{x}))$; 
$w_p = w_q$ may hold. 
	It should be noted that both $\ON(\vec{x}) \cap \bfI_p = \emptyset$ and $\ON(\vec{x}) \cap \bfI_q = \emptyset$ hold, because $w_p, w_q \in \Nvec{\vec{x}}$ are both adjacent with the two vertices in $\ON(\vec{x})$. 
\smallskip

% 	We may assume that $y \in \bfI_{i+1}$, otherwise we can simply drop $\bfI_i$ and obtain a shorter reconfiguration sequence $\langle \bfI_0, \ldots, \bfI_{i-1}, \bfI_{i+1}, \ldots, \bfI_p, \bfI_r \rangle$. 
% 	Thus, we can assume that $i \neq h$. 

	We first claim that $G[\VB(\vec{x})]$ contains an independent set $\bfIB$ such that $|\bfIB| \ge k$ and $\bfIB \cap \Neiclosed{G}{\bfI_p \cup \bfI_q} = \emptyset$.
	By Proposition~\ref{prop:K}, $G[\Nvec{\vec{x}} \cup \ON(\vec{x})]$ contains a complete bipartite graph $K_{|\Nvec{\vec{x}}|,2}$ as a subgraph; recall that $|\ON(\vec{x})| = 2$.
	Therefore, due to the $K_{3,3}$-forbiddance of $G$, every vertex in $V \setminus \ON(\vec{x})$ can be adjacent with at most two vertices in $\Nvec{\vec{x}}$, and hence at most two vertices in $\VB(\vec{x})$.  
	Since both $\ON(\vec{x}) \cap \bfI_p = \emptyset$ and $\ON(\vec{x}) \cap \bfI_q = \emptyset$ hold, 
we have $\bfI_p, \bfI_q \subseteq V \setminus \ON(\vec{x})$ and hence  
	\begin{equation}
	\label{eq:three}
		\bigl|\VB(\vec{x}) \setminus \Neiclosed{G}{\bfI_p \cup \bfI_q} \bigr| \ge \beta - 3 \cdot | \bfI_p \cup \bfI_q | \ge \valbeta - 6 k = 4k. 
	\end{equation}
	By Proposition~\ref{prop:4color} we conclude that $G[\VB(\vec{x}) \setminus \Neiclosed{G}{\bfI_p \cup \bfI_q}]$ contains an independent set $\bfIB$ such that $|\bfIB| \ge k$ and $\bfIB \cap \Neiclosed{G}{\bfI_p \cup \bfI_q} = \emptyset$. 
\smallskip

	We then show that the subsequence $\langle \bfI_{p}, \bfI_{p+1}, \ldots, \bfI_{q} \rangle$ can be replaced with another sequence $\langle \tilde{\bfI}_{p}, \tilde{\bfI}_{p+1}, \ldots, \tilde{\bfI}_{q^\prime} \rangle$ such that $\tilde{\bfI_{j}} \cap \Nvec{\vec{x}} \subseteq \VB(\vec{x})$ for every $j\in \{p, p+1, \ldots,q^\prime\}$. 
	To do so, we use the independent set $\bfIB$ of $G[\VB(\vec{x}) \setminus \Neiclosed{G}{\bfI_p \cup \bfI_q}]$ as a buffer space.
	We now explain how to move the tokens: 
	\begin{listing}{aaa}
	\item[(1)] From the independent set $\bfI_{p-1}$, we first move the token on the vertex $u_p$ in $\bfI_{p-1} \setminus \bfI_{p}$ to an arbitrarily chosen vertex $v^*$ in $\bfIB$, instead of $w_p$.
					Let $\tilde{\bfI}_{p}$ be the resulting vertex set, that is, $\tilde{\bfI}_{p} = (\bfI_{p} \setminus \{ w_p\}) \cup \{ v^* \}$.  
	\smallskip

	\item[(2)] We then move all tokens in $\bfI_{p-1} \setminus \{u_p \}$ $\bigl(= \bfI_{p-1} \cap \bfI_{p} \bigr)$ to vertices in $\bfIB \setminus \{v^* \}$ one by one in an arbitrary order. 
	\smallskip

	\item[(3)] We move tokens in $\bfIB \setminus \{v^* \}$ to the vertices in $\bfI_{q} \cap \bfI_{q+1}$ $\bigl(= \bfI_{q} \setminus \{w_q\} \bigr)$ one by one in an arbitrary order. 
					Let $\tilde{\bfI}_{q^\prime}$ be the resulting vertex set, that is, $\tilde{\bfI}_{q^\prime} = (\bfI_q \setminus \{w_q\}) \cup \{ v^*\} = (\bfI_{q} \cap \bfI_{q+1}) \cup \{ v^* \}$.
	\end{listing}
	Clearly, $\tilde{\bfI_{j}} \cap \Nvec{\vec{x}} \subseteq \VB(\vec{x})$ holds for every $j\in \{p,\ldots,q^\prime\}$. 
	Furthermore, notice that $\bfI_{q+1}$ can be obtained from $\tilde{\bfI}_{q^\prime}$ by moving a single token on $v^* \in \tilde{\bfI}_{q^\prime}$ to the vertex in $\bfI_{q+1} \setminus \bfI_q$. 
	Therefore, $\langle \bfI_{p}, \bfI_{p+1}, \ldots, \bfI_{q} \rangle$ can be correctly replaced with $\langle \tilde{\bfI}_{p}, \tilde{\bfI}_{p+1}, \ldots, \tilde{\bfI}_{q^\prime} \rangle$ if the vertex set $\tilde{\bfI}_{j}$ forms an independent set of $G$ for every $j \in \{ p, \ldots, q^\prime\}$.  
	For every $j \in \{ p, \ldots, q^\prime\}$, notice that either $\tilde{\bfI}_{j} \subset \bfI_p \cup \bfIB$ or $\tilde{\bfI}_{j} \subset \bfI_q \cup \bfIB$ holds. 
	Then, since $\bfIB \cap \Neiclosed{G}{\bfI_p \cup \bfI_q} = \emptyset$, the vertex set $\tilde{\bfI}_{j}$ forms an independent set of $G$. 
	
	By this way, we obtain a reconfiguration sequence $\langle \bfI_0, \ldots, \bfI_{p-1}, \tilde{\bfI}_{p}, \ldots, \tilde{\bfI}_{q^\prime}, \allowbreak \bfI_{q+1}, \ldots, \bfI_r \rangle$ such that no vertex in $\Nvec{\vec{x}} \setminus \VB(\vec{x})$ is contained in any independent set in the sequence. 
\qed
\end{proof}

Lemma~\ref{lem:shrinkreconf} implies that, even if we remove all vertices in $\Nvec{\vec{x}} \setminus \VB(\vec{x})$ for an arbitrary chosen set $\VB(\vec{x}) \subseteq \Nvec{\vec{x}}$ of $\beta = \valbeta$ vertices, it does not affect the existence/nonexistence of a reconfiguration sequence between $\bfI_0$ and $\bfI_r$. 
Thus, we can shrink the subset $\Nvec{\vec{x}}$ into $\VB(\vec{x})$ of size $\beta = \valbeta$. 
\medskip

\noindent
{\bf Step 4}: Line~\ref{alg:s4} of Algorithm \ref{alg}.

In this step, $|\Nvec{\vec{x}}| \le \beta = \valbeta$ hold for all vectors $\vec{x} \in \{0,1\}^{V\setminus \VN}$. 
Furthermore, Proposition~\ref{prop:K} and the $K_{3,3}$-forbiddance of $G$ imply that $|\Nvec{\vec{x}}| \le 2$ if $|\ON(\vec{x})| \ge 3$.
% ; recall the proof of Lemma~\ref{lem:atmosttwo}. 
Since $\alpha = 4k$ and $\beta = \valbeta$, by Eq.~(\ref{eq:step2}) we have
\begin{eqnarray*}
|\VN| 
% &=& \sum \{ |\Nvec{\vec{x}}| : \vec{x} \in \{0,1\}^{V \setminus \VN} \} \nonumber \\
		&=& \sum \{ |\Nvec{\vec{x}}| : \vec{x} \in \{0,1\}^{V \setminus \VN},\ 1 \le |\ON(\vec{x})| \le 2 \} \nonumber \\
			&& ~~ + \sum \{ |\Nvec{\vec{x}}| : \vec{x} \in \{0,1\}^{V \setminus \VN},\ |\ON(\vec{x})| \ge 3 \} \nonumber \\
		&\le& \beta \cdot \left((2k + \alpha) + {2k+\alpha \choose 2} \right) + 2 \cdot \left(2^{2k+\alpha} - (2k + \alpha) - {2k+\alpha \choose 2} \right) \nonumber \\ 
		&=& 2^{6k+1} + 180k^3 - 6k^2 -6k. 
% \label{eq:1}
\end{eqnarray*}
Then, since $|\bfI_0 \cup \bfI_r| \le 2k$ and $|\VR| \le \alpha = 4k$, we can bound $|V|$ by  
	\[
		|V| = |\bfI_0 \cup \bfI_r| + |\VR| + |\VN| < 2^{6k+1} + 180k^3,
	\]
which is denoted by $f_1(k)$.
	Since the order $f_1(k)$ of $G$ now depends only on $k$,  we can apply a brute-force algorithm as follows.
	\begin{lemma} \label{lem:brute}
	If $|V| \le f_1(k)$, {\sc token jumping} is solvable in $O\Bigl( \bigl(f_1(k) \bigr)^{2k} \Bigr)$ time. 
	\end{lemma}
	\begin{proof}
	We construct a \textit{configuration graph} $\cgraph = (\cvertex, \cedge)$, as follows: 
		\begin{listing}{aaa}
		\item[(\one)] each node in $\cgraph$ corresponds to an independent set of $G$ with size $k$; and 
		\item[(\two)] two nodes in $\cgraph$ are joined by an edge if and only if the corresponding two independent sets can be reconfigured by just a single token jump.  
		\end{listing}
	Clearly, there is a reconfiguration sequence between $\bfI_0$ and $\bfI_r$ if and only if there is a path in $\cgraph$ between the two corresponding nodes. 

	Since $G$ has at most the number ${f_1(k) \choose k}$ of distinct independent sets of size exactly $k$, we have $|\cvertex| \le \bigl( f_1(k) \bigr)^k$. 
	The configuration graph $\cgraph$ above can be constructed in $O(|\cvertex|^2)$ time. 
	Furthermore, by the breadth-first search on $\cgraph$ which starts from the node corresponding to $\bfI_0$, we can check if $\cgraph$ has a desired path or not in $O(|\cvertex| + |\cedge|) = O(|\cvertex|^2)$ time. 
	In this way, {\sc token jumping} can be solved in $O(|\cvertex|^2) = O\bigl( \bigl(f_1(k) \bigr)^{2k} \bigr)$ time in total. 
	\qed
	\end{proof}

This completes the correctness proof of Algorithm~\ref{alg}. 
\medskip

\noindent
{\bf Running time.}

We now estimate the running time of Algorithm~\ref{alg}.
We first claim that lines~\ref{alg:s0}--\ref{alg:f5} can be executed in $O(|E|)$ time.  
Lines~\ref{alg:sfor}--\ref{alg:ffor} can be clearly done in fixed-parameter running time, but actually these lines can be done in $O(|E|)$ time because $|\VN| = |\bigcup \{ \Nvec{\vec{x}} : \vec{x} \in \{0,1\}^{V \setminus \VN} \}|$ is at most $n$; 
we can compute $\vec{x}$ implicitly. 
By Lemma~\ref{lem:brute} we can execute line~\ref{alg:s4} in $O\bigl( \bigl(f_1(k) \bigr)^{2k} \bigr)$ time.
Thus, the total running time of Algorithm~\ref{alg} is $O\bigl(|E|+\bigl(f_1(k) \bigr)^{2k} \bigr)$.  

This completes the proof of Theorem~\ref{the:planar}.

\subsection{$K_{3,t}$-forbidden graphs} 
\label{subsec:extension}

	In this subsection, we show that our algorithm for planar graphs can be extended to that for $K_{3,t}$-forbidden graphs, and give the following theorem. 
	\begin{theorem} \label{the:extension}
	For a fixed integer $t \ge 3$, let $G$ be a $K_{3,t}$-forbidden graph. 
	Then, {\sc token jumping} for $G$ can be solved in fixed-parameter running time, when parameterized by the number $k$ of tokens. 
	\end{theorem}

	We here give a sketch of how to adapt the fixed-parameter algorithm for planar graphs in Section~\ref{subsec:planar} to $K_{3,t}$-forbidden graphs. 
	
	The first point is to set two parameters $\alpha_t$ and $\beta_t$ that correspond to $\alpha$ and  $\beta$, respectively. 
	Recall that $\alpha = 4k$ and $\beta = \valbeta$ are the orders of (sub)graphs that guarantee the existence of sufficiently large independent sets that will be used as a buffer space. 
	For $K_{3,t}$-forbidden graphs, we set $\alpha_t = \Ram(k, t+3)$ and $\beta_t = \Ram((2t+1)k, t+3)$.
	Then, Proposition~\ref{prop:free} guarantees the existence of independent sets of size $k$ in $\VR$, and hence Step~1 of Algorithm~\ref{alg} can be adapted to $K_{3,t}$-forbidden graphs. 
	We note that, although no exact formula of Ramsey number is known, we can bound it from above, say $\Ram(a,b)\le {a+b-2 \choose b-1}$~\cite{Graham:1990:RT:85271}. 
	Therefore, we indeed set $\alpha_t = (k+t+1)^{t+2}$ and  $\beta_t = \bigl((2t+1)k+t+1 \bigr)^{t+2}$, both of which are fixed-parameter size.

	The second point is to extend Lemma~\ref{lem:atmosttwo} for planar graphs to that for $K_{3,t}$-forbidden graphs, as follows. 
	(Note that $\beta_t \ge t$.)
	\begin{lemma} \label{lem:atmosttwoK3t}
	If $|\Nvec{\vec{x}}| \ge \beta_t$ holds for a vector $\vec{x} \in \{0,1\}^{V\setminus \VN}$, then $|\ON(\vec{x})| \le 2$. 
	\end{lemma}
	Then, since there is an independent set of size at least $(2t+1)k$ in $\Nvec{\vec{x}}$, Step~2 of Algorithm~\ref{alg} can be adapted to $K_{3,t}$-forbidden graphs. 
	
	The third point is to modify Eq.~(\ref{eq:three}) in the proof of Lemma~\ref{lem:shrinkreconf} and shrink $\Nvec{\vec{x}}$ to size $\beta_t$. 
	Recall that the vertices in $\ON(\vec{x})$ and $\Nvec{\vec{x}}$ form a complete bipartite graph $K_{2, |\Nvec{\vec{x}}|}$, and hence any vertex other than $\ON(\vec{x})$ can be adjacent with at most $(t-1)$ vertices in $\Nvec{\vec{x}}$, due to the $K_{3,t}$-forbiddance of the graph. 
	Therefore, if $\Nvec{\vec{x}}$ has an independent set of size at least $(2t+1)k$, we still have at least $k$ vertices that can be used as a buffer space. 
	Therefore, Lemma~\ref{lem:shrinkreconf} can be adapted to $K_{3,t}$-forbidden graphs, and hence Step~3 of Algorithm~\ref{alg} can be, too.
	
	The running time of the adapted algorithm depends on the order of the graph shrunk by Step~3.
	By the similar arguments for planar graphs, the order of the shrunk graph depends only on $\alpha_t$ and $\beta_t$. 
	Since both $\alpha_t$ and $\beta_t$ are fixed-parameter size, the adapted algorithm runs in fixed-parameter running time. 
	
	This completes the proof of Theorem~\ref{the:extension}.
% \qed

\section{Shortest Reconfiguration Sequence}

	In the previous section, we present an algorithm which simply determines if there exists a reconfiguration sequence between two given independent sets $\bfI_0$ and $\bfI_r$. 
	If the answer is yes, it is natural to consider how we actually move tokens on $\bfI_0$ to $\bfI_r$. 
	For this question, it is easy to modify Algorithm~\ref{alg} to output an actual reconfiguration sequence, although it is not always shortest.
% by simply following the correctness proof of Algorithm~\ref{alg}.
%
\if0
	Note that, in Steps~1 and 2 of Algorithm~\ref{alg}, we utilize the existence of a sufficiently large independent set, but if we want to find a reconfiguration sequence, we need to find a concrete independent set of a certain size. 
	Namely, we need to find
(1) an independent set of size $k$ from a graph $G[\VR]$ of order at least $\alpha$ at Step~1; and 
(2) an independent set of size $\valcli$ (or size $(t-1)(2k-1)+k$ for $K_{3,t}$-free graphs) from a graph $G[\Nvec{\vec{x}}]$ of order at least $\beta$ at Step~2. 
	Notice that we can use any independent set as a buffer space, as long as the size requirement is satisfied.
	Therefore, we choose an arbitrary vertex subset of size $\alpha$ from $\VR$ (or size $\beta$ from $\Nvec{\vec{x}}$) and find independent sets in a brute-force manner. 
	These require simply $O \bigl({\alpha \choose k} \bigr)$ time and $O \bigl({\beta \choose \valcli} \bigr)$ time (or $O \bigl({\beta_t \choose (t-1)(2k-1)+k} \bigr)$ time for $K_{3,t}$-free graphs), which are still fixed-parameter running time. 
	Actually, for planar graphs, $O(\alpha^2)$ time  and $O(\beta^2)$ time are sufficient by Proposition \ref{prop:4color}.  
% \medskip
\fi
%
	In this section, we consider how to move tokens on $I_0$ to $I_r$ in a shortest way. 
% 	This is not obvious compared with the previous question, but we show the following theorem. 
	\begin{theorem} \label{the:shortest}
	For a fixed integer $t \ge 3$, let $G$ be a $K_{3,t}$-forbidden graph. 
	Given a yes-instance of {\sc token jumping} on $G$, a shortest reconfiguration sequence can be found in fixed-parameter running time, where the parameter is the number $k$ of tokens.
	\end{theorem}

\noindent
{\em Proof sketch.}
% 	\begin{proof}
	We explain how to modify Algorithm~\ref{alg} so as to find a shortest reconfiguration sequence. 
	The biggest change from Algorithm~\ref{alg} is that the modified algorithm does not stop until Step~4. 
	Algorithm~\ref{alg} can exit at Steps~1 and 2 after finding a buffer space, which means that there exists a reconfiguration sequence from $\bfI_0$ to $\bfI_r$ via vertices only in $\VR$ and vertices only in $\Nvec{\vec{x}}$, respectively. 
	However, 
% the existence of a reconfiguration sequence that uses only $A(\vec{x})$ and $I_0\cup I_r$ 
this does not directly imply the existence of a {\em shortest} reconfiguration sequence from $\bfI_0$ to $\bfI_r$ that uses vertices only in $\VR$ (or only in $\Nvec{\vec{x}}$). 
	Thus, we do not exit at Steps~1 and 2, but shrink $\VR$ and $\Nvec{\vec{x}}$ of the original graph into a fixed-parameter size so as to preserve the shortest length of a reconfiguration sequence in the original graph;
then we can find a shortest reconfiguration sequence in Step~4 by the brute-force algorithm proposed in Lemma~\ref{lem:brute}.  
(Details are omitted from this extended abstract.)
	\qed
% 	\end{proof}

% 	\section{Concluding Remark}

% 	The running time of our algorithms depend on the orders of (sub)graphs that guarantee the existence of sufficiently large independent sets.
% 	Thus, we can improve the running time if the graph class is restricted to $K_{3,t}$-minor-free graphs for which smaller orders are sufficient to guarantee large independent sets. 

	\section*{Acknowledgments}
	We are grateful to Akira Suzuki, Ryuhei Uehara, and Katsuhisa Yamanaka for fruitful discussions with them. 
	We thank anonymous referees for their helpful suggestions.
	This work is partially supported by JSPS KAKENHI 25106504 and 25330003 (T.~Ito), and 25104521 and 26540005 (H.~Ono), and by the (Polish) National Science Center under grant no.  2013/09/B/ST6/03136 (M. Kami\'nski).

%\bibliographystyle{splncs03}
%\bibliography{token}
\bibliographystyle{abbrv}

\end{document}